\documentclass[12pt,authoryear,review]{elsarticle}

\makeatletter
\def\ps@pprintTitle{%
 \let\@oddhead\@empty
 \let\@evenhead\@empty
 \def\@oddfoot{}%
 \let\@evenfoot\@oddfoot}
\makeatother

\usepackage{color}
\usepackage{booktabs}
\usepackage{mathtools}
\usepackage{amsmath,amssymb,amsfonts}
\usepackage{mathrsfs}
\usepackage{graphics}
\usepackage{graphicx}
\usepackage{dsfont}
\usepackage{subfigure}
\usepackage{url}
\usepackage[labelfont=bf,font={footnotesize,it},format=hang,width=0.8\textwidth]{caption}
\usepackage[titletoc,title]{appendix}
\usepackage{enumitem}
\usepackage{rotating}
\usepackage{pdflscape}
\usepackage{afterpage}

\setitemize{itemsep=0pt}
\numberwithin{equation}{section}
\newtheorem{Theorem}{Theorem}
\newtheorem{Lemma}{Lemma}
\newtheorem{Proposition}{Proposition}
\newtheorem{Corollary}{Corollary}
\newtheorem{Remark}{Remark}
\newtheorem{Assumption}{Assumption}
\newtheorem{Definition}{Definition}
\newenvironment{proof}[1][Proof]{\textbf{#1.} }{\ \rule{0.5em}{0.5em}}
\renewcommand{\theequation}{\thesection.\arabic{equation}}
\newcommand\numberthis{\addtocounter{equation}{1}\tag{\theequation}}

\newcommand{\EE}{{\mathbb E}}
\newcommand{\NN}{{\mathbb N}}

\newcommand{\RR}{{\mathbb R}}
\newcommand{\PP}{{\mathbb P}}
\newcommand{\LL}{{\mathbb L}}
\newcommand{\ind}{\mathds{1}}
\newcommand{\As}{{\mathcal A}}

\newcommand{\Ms}{{\mathcal M}}

\newcommand{\Ft}{{\mathcal F}}
\newcommand{\FF}{{\mathfrak F}}

\newcommand{\ones}{\mathbf{1}}
\newcommand{\zeros}{\mathbf{0}}
\newcommand{\vb}{{\boldsymbol b}}
\newcommand{\vg}{{\boldsymbol g}}

\newcommand{\vx}{{\boldsymbol x}}

\newcommand{\vM}{{\boldsymbol M}}

\newcommand{\vpi}{{\boldsymbol \pi}}

\newcommand{\valpha}{{\boldsymbol \alpha}}
\newcommand{\vtheta}{{\boldsymbol \theta}}

\newcommand{\vsigma}{{\boldsymbol \sigma}}

\newcommand{\vomega}{{\boldsymbol \omega}}

\newcommand{\mSigma}{\mathbf \Sigma}

\newcommand{\mx}{{\mathbf{x}}}

\newcommand{\mfN}{{\mathfrak{N}}}

\graphicspath{{figures/}}

\usepackage{array}
\newcolumntype{L}[1]{>{\raggedright\let\newline\\\arraybackslash\hspace{0pt}}m{#1}}
\newcolumntype{C}[1]{>{\centering\let\newline\\\arraybackslash\hspace{0pt}}m{#1}}
\newcolumntype{R}[1]{>{\raggedleft\let\newline\\\arraybackslash\hspace{0pt}}m{#1}}

\begin{document}


\begin{frontmatter}

\setlength{\parskip}{0em}

\title{\textbf{A Variational Analysis Approach to Solving the Merton Problem}
\\[1em]}
\tnotetext[t1]{SJ and AA would like to acknowledge the support of the Natural Sciences and Engineering Research Council of Canada (NSERC). For SJ, under the funding reference numbers RGPIN-2018-05705 and RGPAS-2018-522715. }

\author{Ali Al-Aradi}
\ead{ali.al.aradi@utoronto.ca}

\author{Sebastian Jaimungal}
\ead{sebastian.jaimungal@utoronto.ca}
\address{Department of Statistical Sciences, University of Toronto}

\begin{abstract}
We address the Merton problem of maximizing the expected utility of terminal wealth using techniques from variational analysis. Under a general continuous semimartingale market model with stochastic parameters, we obtain a characterization of the optimal portfolio for general utility functions in terms of a forward-backward stochastic differential equation (FBSDE) and derive solutions for a number of well-known utility functions. Our results complement a previous study conducted in \cite{ferland2008fbsde} on optimal strategies in markets driven by Brownian noise with random drift and volatility parameters.
\end{abstract}

\begin{keyword}
Merton problem;
Portfolio selection;
Stochastic control;
Convex analysis;
Variational analysis.
\end{keyword}

\end{frontmatter}

\section{Introduction}

The Merton problem is among the most well-known and well-studied problems in mathematical finance. Introduced in the seminal works of \cite{Merton69} and \cite{Merton71}, the problem is one of dynamic asset allocation and consumption in which an investor chooses to allocate their wealth between a risk-free asset and a risky asset with the goal of maximizing expected utility of consumption and terminal wealth. Although the two initial papers consider a number of variations to the problem such as
stochastic additions to wealth other than capital gains (e.g. wages), making the risk-free asset defaultable and using alternatives to geometric Brownian motion for modeling asset price behavior, the papers still managed to spawn numerous extensions in other directions. For example, the incorporation of transaction costs in \cite{Magill76} and \cite{Davis90}, uncertain investment horizon in \cite{BlanchetScaillet2008}, taxes to capital gains in \cite{tahar2010merton} and illiquid assets in \cite{Ang2014}. A number of works also consider variations of the Merton problem with partial information, e.g. \cite{bauerle2007} and \cite{frey2012}. The Merton problem also plays a role when utility maximization is not the direct goal, as is the case with indifference pricing; see \cite{henderson2002valuation} for the valuation of claims on non-traded assets and \cite{henderson2004utility} for a broader survey of the topic.

The Merton problem is also interesting due to the multitude of approaches adopted for solving the different variations of the problem. The original works of \cite{Merton69} and \cite{Merton71}, along with the majority of subsequent papers, tackle the problem via dynamic programming and solving the pertinent Hamilton-Jacobi-Bellman equations. Other techniques that have been adopted include the use of martingale methods and duality theory to solve the problem in incomplete financial markets as in \cite{karatzas1991martingale}; see also \cite{schachermayer2002optimal} for a survey of this topic. \cite{cvitanic1992convex} uses similar techniques to solve a constrained  version of the portfolio optimization problem. More recently, \cite{ferland2008fbsde} utilizes the stochastic Pontryagin maximum principle to characterize the optimal portfolio in terms of a forward-backward stochastic differential equation (FBSDE) for a market with random parameters driven by Brownian noise and a general utility function. \cite{rieder2012robust} apply the same technique to solving a robust version of the same problem where utility maximization is performed under the worst-case parameter configuration.

In this paper we focus on the investment portion of the Merton problem where the investor seeks to maximize the expected utility of terminal wealth. The approach we take to solve the utility maximization problem is based on techniques from variational analysis as discussed in, for example, \cite{ekeland1999convex}. Similar techniques were applied by \cite{bank2019optimal} to solve the Merton problem in the presence of price impact, as well as \cite{bank2017hedging} and \cite{casgrain2018algorithmic} in the contexts of hedging contingent claims in the presence of price impact and a mean field game approach to optimal execution, respectively. To our knowledge, the simple version of the Merton problem has not been tackled with this approach, which brings new insight to this well-studied problem and may be useful for solving other related problems, including the extension of \cite{alaradi2019active} to incorporate general utility functions.

\section{Model Setup} \label{sec:modelSetup}

\subsection{Market Model}

Let $(\Omega, \Ft, \FF, \PP)$ be a filtered probability space, where $\FF = \{\Ft_t\}_{t \geq 0}$ is the natural filtration generated by all processes in the model. We  assume that the market consists of $n$ risky assets and one risk-free asset which are defined as follows:
\begin{Definition}
The \textbf{stock price process} for risky asset $i$, $P^i = \left( P^i_t \right)_{t \geq 0}$ for all $i\in\mfN \coloneqq \{1,\dots,n\}$, is a positive semimartingale satisfying the stochastic differential equation (SDE):
\begin{equation}
dP^i_t =  \alpha^i_t P_t^i \, dt + P_t^i \, dM^i_t \, , \qquad P^i_0 = p^i \, ,
\end{equation}
where $\alpha^i = \left( \alpha^i_t \right)_{t \geq 0}$ is an $\FF$-adapted process representing the asset's \textbf{instantaneous rate of return} and $M^i = \left( M^i_t \right)_{t \geq 0}$ is an $\FF$-martingale with $M^i_0 = 0$ representing the asset's \textbf{noise component}. The \textbf{risk-free asset price process}, $P^0 = \left( P^0_t \right)_{t \geq 0}$ is a positive semimartingale satisfying the SDE
\begin{equation}
dP^0_t = r_t P_t^0 \, dt \, , \qquad P^0_0 = p^0 \, ,
\end{equation}
where $r = \left( r_t \right)_{t \geq 0}$ is an $\FF$-adapted process representing the risk-free rate.
\end{Definition}

Next, we specify the assumptions made on the various market model processes. To this end, we first define the spaces
{ \small
\begin{align}
\LL_T^p(\RR^n) &= \left\{ f: \Omega \times [0,T] \rightarrow \RR^n \text{ s.t. }  \EE \left[ \int_{0}^{T} \left(\| f_t \|_p\right)^p ~dt \right] < \infty \right\}\,, \quad 0<p<\infty \, ,
\\
\text{and } ~ \LL_T^{\infty,M}(\RR^n) &= \left\{ f: \Omega \times [0,T] \rightarrow \RR^n \text{ s.t. }  \sup_{t\in[0,T]}\| f_t \|_\infty \leq M, \; \PP-a.s. \right\} \, ,
\end{align}
}%
where $\| \vx \|_p \coloneqq \left(\sum_{i=1}^n |x_i|^p \right)^{1/p}$ and $\| \vx \|_\infty \coloneqq \underset{i \in \mfN}{\max} \hspace{0.1cm} |x_i|$ for $\vx \in \RR^n$  denote the $p$-norm and $\infty$-norm on $\RR^n$, respectively. Furthermore, we will make use of the shorthand notation $\| \vx \| \coloneqq \| \vx \|_2$ to denote the usual Euclidean norm.
\begin{Assumption} \label{asmp:growthMtg}
	The risk-free rate and rate of return processes are continuous and bounded, i.e. $r \in \LL_T^{\infty,M}(\RR)$, $ \valpha \in \LL_T^{\infty,M}(\RR^n)$, where $\valpha_t = \left(\alpha^1_t, ..., \alpha^n_t \right)^\intercal$. Additionally, the martingale noise processes are assumed to be continuous with finite second moments, i.e. $\EE \left( \| \vM_t \|^2 \right) < \infty$ for all $t \geq 0$.
\end{Assumption}

We also assume that the quadratic co-variation processes associated with the noise component satisfy
\begin{Assumption} \label{asmp:posDef}
Let $\mSigma$ be the matrix whose $ij$-th element is the \textbf{quadratic covariation} process between $M^i$ and $M^j$, $\mSigma^{ij}_t \coloneqq \langle M^i, M^j \rangle_t$. We assume that, for each $\mx \in \RR^n$, there exists $\varepsilon > 0$ and $C < \infty$    such that
\begin{align}
\varepsilon \| \mx \|^2 ~\leq~ \mx^\intercal \mSigma_t \mx  ~\leq~ C \| \mathbf{x} \|^2\,, \qquad
\forall t\ge0.
\end{align}
\end{Assumption}
This is an extension of the usual \textbf{non-degeneracy} and \textbf{bounded variance} conditions. Note that since $\vM$ is continuous, $\mSigma$ is a continuous process as well.


Next, we define portfolio processes which will constitute the investor's control in the optimization problem.
\begin{Definition}
A \textbf{portfolio} is an $\FF$-predictable, vector-valued process $\vpi = ( \vpi_t )_{t \geq 0}$, with $\vpi \in \LL^2_T(\RR^n)$ where $\vpi_t = \left(\pi^1_t, ..., \pi^n_t \right)^\intercal$ such that for all $t \geq 0$, $\pi_t^i$ represents the proportion of wealth invested in risky asset $i$ and $\pi_t^0 = 1 - \pi^1_t - \cdots - \pi^n_t$ is the proportion invested in the risk-free asset. We denote the set of all portfolios by
\begin{equation}
\As = \Big\{ \vpi:\Omega \times [0,T] \rightarrow \RR^n \text{ s.t. } \vpi\in\LL_T^2(\RR^n),~  \FF\text{-predictable} \Big\} \, .
\end{equation}
\end{Definition}

Given the model dynamics and portfolio assumptions, the \textbf{portfolio value process} $X^\vpi = (X^\vpi_t )_{t \geq 0}$ associated with an arbitrary portfolio $\vpi$  satisfies the SDE
\begin{equation} \label{eqn:portSDE}
d X^\vpi_t  = X^\vpi_t \big(r_t + \vpi_t^\intercal \vtheta_t \big) \, dt + X^\vpi_t \vpi_t^\intercal \, d\vM_t\,, \qquad X_0^\vpi = x > 0 \, ,
\end{equation}
where $\vtheta_t = \valpha_t - r_t \ones$ is the vector of excess returns, $\ones$ is a vector of ones and $x$ is the investor's initial wealth. It will also be convenient at times to work with the logarithm of wealth, which satisfies the SDE
\begin{align} \label{eqn:logWealth}
d \log X_t^\vpi = \gamma_t^\vpi \, dt + \vpi_t^\intercal \, d\vM_t \, ,
\end{align}
where $\gamma^{\vpi}_t = r_t + \vpi_t^\intercal \vtheta_t - \tfrac{1}{2} \vpi_t^\intercal \mSigma_t \vpi_t$ is the \textbf{portfolio growth rate}.

Due to some technical requirements that will come into play when solving the optimization problem, we initially restrict ourselves to strategies where wealth is transferred to the risk-free asset for the remainder of the investment horizon once a certain level of wealth is reached. To formalize this restriction, we first define the \textbf{$K$-stopped version} of $\vpi$ as the portfolio $\vpi_K = ( \vpi_{K,t} )_{t \geq 0}$ given by
\begin{equation*}
\vpi_{K,t} =
\begin{cases}
\vpi_t \, , & ~~ \text{\textit{if} } \left| X_t^\vpi \right| \leq K
\\
\zeros \, , & ~~ \text{\textit{if} } \left| X_t^\vpi \right| > K
\end{cases}
\end{equation*}
We also define the stopping time associated with reaching the wealth threshold, namely
\begin{equation}
\tau^{\vpi_K} = \inf \big\{ t \geq 0 : | X_t^\vpi | > K \big\} \, ,
\end{equation}
along with the associated indicator process $\ind^{\vpi_K} = \left( \ind^{\vpi_K}_t \right)_{t \geq 0} $ defined as
\begin{equation}
\ind^{\vpi_K}_t = \ind \left\{ \tau^{\vpi_K} \, > \, t \right\}  \, .
\end{equation}
Finally, we define the set of constrained portfolios for a given wealth threshold.
\begin{Definition}
The set of \textbf{admissible portfolios for the $K$-constrained problem} consists of those strategies that are stopped once the wealth threshold $K$ is reached, denoted
\begin{equation}
\As_K = \big\{ \vpi_K  : \,   \vpi \in \As \big\} \, .
\end{equation}
\end{Definition}

\section{Stochastic Control Problem} \label{sec:control}

The stochastic control problem we consider is the Merton problem without consumption. More specifically, the investor's objective is to determine the portfolio process $\vpi \in \As$ that maximizes their expected utility of terminal wealth at the end of their investment horizon $T$. In mathematical terms our \textbf{stochastic control problem} is to find the optimal portfolio $\vpi^*$ which, if the supremum is attained in the set of admissible strategies, achieves
\begin{equation} \label{eqn:control}
\underset{\vpi \in \As}{\sup} ~ H(\vpi) \, ,
\end{equation}
where $H: \As \rightarrow \RR$ is the \textbf{performance criteria} of an admissible portfolio $\vpi \in \As$ given by
\begin{align} \label{eqn:perfCrit}
H(\vpi) \coloneqq \EE \left[ U(X_T^\vpi) \right]\;.
\end{align}

We approach solving \eqref{eqn:control} by solving a sequence of nested constrained problems where the search space is reduced to $\As_K$, namely
\begin{equation} \label{eqn:constrained_control}
\underset{\vpi_K \in \As_K}{\sup} \, H_K(\vpi_K) \, ,
\end{equation}
where $H_K: \As_K \rightarrow \RR$ is the performance criteria of a $K$-stopped portfolio $\vpi_K \in \As_K$ given by
\begin{align} \label{eqn:constrained_perfCrit}
H_K(\vpi_K) \coloneqq \EE \left[ U(X_T^{\vpi_K}) \right]\;.
\end{align}

In the expressions above, $U$ is a von Neumann-Morgenstern utility function which reflects the investor's preferences and satisfies
\begin{Assumption}
The investor's utility function $U$ is three times continuously differentiable, increasing and strictly concave, i.e. $U^{(1)}(x) > 0$ and $U^{(2)}(x) < 0$ for all $x > 0$, where $U^{(k)}$ is the $k$th derivative of $U$.
\end{Assumption}
For convenience we will define the \textbf{utility process} $Z^\vpi = (Z^\vpi_t)_{t\geq0}$ as
\begin{equation}
Z_t^\vpi \coloneqq U(X_t^\vpi) \,.
\end{equation}

We proceed with solving the optimal control problem in four parts: (i) we establish the existence and uniqueness of a global optimizer for the stochastic control problems \eqref{eqn:control} and \eqref{eqn:constrained_control}; (ii) we compute the G\^{a}teaux derivative associated with the functional $H_K$; (iii) we find an element in the admissible set $\As_K$ which makes the derivative vanish and relate it to the solution of a FBSDE; (iv) we take the limit as $K$ tends to infinity to obtain an expression for $\vpi^*$.

For the remainder of the paper we address the constrained problem \eqref{eqn:constrained_control} with a fixed $K$ unless explicitly stated and we will omit the subscripts from the control processes for notational convenience.

\subsection{Existence and Uniqueness of a Global Maximum}

To show the existence and uniqueness of a global optimizer for \eqref{eqn:constrained_control}, we use the strict concavity of a related control problem that uses the dollar amount process as the investor's control and demonstrate a one-to-one correspondence between the control processes of the two problems. An analogous argument can be used to prove a similar result for \eqref{eqn:control}, which we omit.
\begin{Proposition} \label{prop:concave}
The stochastic control problem \eqref{eqn:constrained_control} has a unique global maximizer.
\end{Proposition}
\begin{proof}
We define an auxiliary control problem with the same performance criteria as \eqref{eqn:constrained_perfCrit} but where the control process is a vector of \textit{dollar amounts}, rather than \textit{proportions} of wealth, invested in each asset and denote the new control process by $\widetilde{\vpi} = ( \widetilde{\vpi}_t )_{t \geq 0}$ where $\widetilde{\vpi}_t = \left(\widetilde{\pi}^1_t, ..., \widetilde{\pi}^n_t \right)^\intercal$. In particular, we are interested in
\begin{equation} \label{eqn:auxControl}
\underset{\widetilde{\vpi} \in \As_K^*}{\sup} ~ J(\widetilde{\vpi}) \, ,
\end{equation}
where $J: \As_K^* \rightarrow \RR$ is given by
\[
J(\widetilde{\vpi}) \coloneqq \EE \left[ U \big( X_T^{\widetilde{\vpi}} \big) \right]
\]
and $\As_K^*$ is the set of admissible portfolios expressed in terms of dollar amounts. The control processes in the two optimization problems are related via
\[
\vpi^i_t = \frac{\widetilde{\pi}^i_t}{\widetilde{\pi}^1_t + \cdots + \widetilde{\pi}^n_t} \qquad \text{for } i \in \mathfrak{N} \, ,
\]
or through the wealth process as follows
\[
X_t^\vpi \vpi_t = \widetilde{\vpi}_t \, .
\]
Given a fixed initial wealth and using the fact that portfolios are self-financing, there exists a one-to-one mapping between $\As_K$ and $\As_K^*$. Additionally, the numerical value of the two functionals $H$ and $J$ are equal when taking two controls that map to one another. This implies that if the auxiliary control problem has a unique global maximizer then so does the control problem \eqref{eqn:constrained_control}.

To show that the auxiliary control problem has a unique solution we show that the functional $J$ is strictly concave in the control $\widetilde{\vpi}$ and that the search space $\As_K^*$ is a convex set. This is done in two separate lemmas.
\begin{Lemma}
	The functional $J$ is strictly concave in the dollar amount process $\widetilde{\vpi}$.
\end{Lemma}
\begin{proof}
The wealth process controlled via the dollar amounts follows the dynamics
\begin{equation} \label{eqn:dollarSDE}
d X^{\widetilde{\vpi}}_t  = \big( r_t X^{\widetilde{\vpi}}_t + \widetilde{\vpi}_t^\intercal \vtheta_t \big) \, dt + \widetilde{\vpi}_t^\intercal \, d\vM_t\,, \qquad X_0^{\widetilde{\vpi}} = x > 0 \, .
\end{equation}
This process can be linearized by defining $\kappa_{0,t} X^{\widetilde{\vpi}}_t$ where $\kappa_{s,t} = e^{-\int_s^t r_u du}$. This process satisfies
\begin{align*}
d \left( \kappa_{0,t} X^{\widetilde{\vpi}}_t \right) &= \kappa_{0,t} \, d X^{\widetilde{\vpi}}_t - r_t \kappa_{0,t} X^{\widetilde{\vpi}}_t \, dt
\\
&= \kappa_{0,t} \widetilde{\vpi}_t^\intercal \vtheta_t \, dt + \kappa_{0,t} \widetilde{\vpi}_t^\intercal \, d\vM_t \, .
\end{align*}
Integrating we obtain an expression for $X^{\widetilde{\vpi}}_T$
\[
X^{\widetilde{\vpi}}_T = x \, \kappa_{T,0} + \int_0^T \kappa_{T,t} \widetilde{\vpi}_t^\intercal \vtheta_t \, dt + \int_0^T \kappa_{T,t} \widetilde{\vpi}_t^\intercal \,  d\vM_t \,.
\]
Since this expression is linear in $\widetilde{\vpi}$, for any $\widetilde{\vpi}, \widetilde{\vomega} \in \As_K^*$ and $c \in [0,1]$ we have
\[
X^{c\widetilde{\vpi} + (1-c) \widetilde{\vomega}}_T = c X^{\widetilde{\vpi}}_T + (1-c)  X^{\widetilde{\vomega}}_T \, ,
\]
and by the strict concavity of $U$ it follows that
\[
U \left( X^{c\widetilde{\vpi} + (1-c) \widetilde{\vomega}}_T \right) > c ~ U ( X^{\widetilde{\vpi}}_T ) + (1-c) ~ U ( X^{\widetilde{\vomega}}_T ) \, .
\]
Taking expectations establishes the strict concavity of the functional $J$.
\end{proof}
\begin{Lemma} \label{lemma:convex_set}
	The search space $\As_K^*$ for the auxiliary control problem \eqref{eqn:auxControl} is a convex set.
\end{Lemma}
\begin{proof}
	Fix $c \in [0,1]$ and $\widetilde{\vpi}_1, \widetilde{\vpi}_2 \in \As_K^*$ and consider the convex combination $\widetilde{\vomega} = c \widetilde{\vpi}_1 + (1-c) \widetilde{\vpi}_2$. The goal is to show that this is an admissible dollar amount process. Clearly, $\widetilde{\vomega}$ is $\FF$-predictable and corresponds to a portfolio process in $\As_K$, so all that remains is to show that it is stopped once $X^{\widetilde{\vomega}}$ hits the wealth threshold $K$.
	
	Recall from the proof of the previous lemma that $X^{\widetilde{\vomega}} = c X^{\widetilde{\vpi}_1} + (1-c) X^{\widetilde{\vpi}_2}$. Moreover, if a subportfolio $\widetilde{\vpi}_i$ is not stopped then the absolute value of its associated wealth is necessarily less than the threshold $K$. With this in mind, there are three cases to consider: neither of the subportfolios is stopped, exactly one subportfolio is stopped or both subportfolios are stopped. In the first two cases it can be easily verified that $|X^{\widetilde{\vomega}}| \leq K$ and in the last case $ | X^{\widetilde{\vomega}} | = K$ and the portfolio $\widetilde{\vomega}$ is stopped since both $\widetilde{\vpi}_1$ and $\widetilde{\vpi}_2$ are stopped. Therefore, $\widetilde{\vomega} \in \As_K^*$ and hence $\As_K^*$ is a convex set.
\end{proof}

Since the functional $J$ is strictly concave and the set $\As_K^*$ is convex, the auxiliary control problem \eqref{eqn:auxControl} has a unique global maximizer and therefore so does the constrained control problem \eqref{eqn:constrained_control}. This completes the proof.
\end{proof}

\subsection{The G\^{a}teaux derivative}

In this section we derive an expression for the G\^{a}teaux derivative of the functional $H_K$ in a number of incremental steps, starting with the following lemma:
%
%
\begin{Lemma} \label{lemma:perturb}
Fix $\epsilon > 0$ and two portfolio processes $\vpi, \vomega \in \As_K$ and define the following processes:
\begin{subequations}
\begin{align}
F_t^{\vpi, (k)} &\coloneqq U^{(k)}(X_t^{\vpi }) (X_t^\vpi)^k ~~ \text{ for } k = 1,2,3 \, ,
\\
I^\vomega_t &\coloneqq \int_0^t \vomega_u^\intercal \left( \vtheta_u - \mSigma_u \vpi_u \right) \, du + \int_0^t \vomega_u^\intercal \, d\vM_u \label{eqn:integral} \, ,
\\
\vg^\vpi_t &\coloneqq F_t^{\vpi,(1)} \vtheta_t + F_t^{\vpi,(2)} \mSigma_t \vpi_t \, ,
\\
h^\vpi_t &\coloneqq \Big( F_t^{\vpi,(1)} + F_t^{\vpi,(2)} \Big) (r_t + \vpi_t^\intercal  \vtheta_t) + \Big( F_t^{\vpi,(2)} + \tfrac{1}{2} F_t^{\vpi,(3)} \Big) \vpi_t^\intercal \mSigma_t \vpi_t \, .
\end{align}
\end{subequations}
Then we have
\begin{equation} \label{eqn:gateaux1}
H_K(\vpi + \epsilon \vomega) = \EE \left[ Z^\vpi_{T \wedge \tau^{\vpi + \epsilon \vomega} }\right] + \epsilon ~ \EE \left[ \int_0^{ T \wedge \tau^{\vpi + \epsilon \vomega} } \Big( \vomega_t^\intercal \vg^\vpi_t + I_t^\vomega h^\vpi_t \Big) \, dt \right] + o(\epsilon) \,.
\end{equation}
\end{Lemma}
\begin{proof}
First, using \eqref{eqn:portSDE} and the definition of $F_t^{\vpi, (k)}$, we obtain the dynamics of $Z_t^\vpi$ by applying It\^{o}'s lemma which gives
{ \small
\begin{align} \label{eqn:utilitySDE}
dZ_t^\vpi &= \Big( F_t^{\vpi, (1)} \big(r_t + \vpi_t^\intercal \vtheta_t \big) + \tfrac{1}{2} F_t^{\vpi, (2)} \vpi_t^\intercal \mSigma_t \vpi_t \Big) \, dt + F_t^{\vpi, (1)} \vpi_t^\intercal \, d\vM_t \, .
\end{align} }%
We are interested in the dynamics of the perturbed utility process $Z^{\vpi + \epsilon \vomega}$, i.e. the utility process induced by the control $\vpi + \epsilon \vomega$. It is important to note that in order for $\vpi + \epsilon \vomega$ to be an admissible portfolio it must be stopped in the usual manner once its associated wealth process reaches the threshold $K$. To arrive at the SDE satisfied by this process, we first perturb the growth rate $\gamma^\vpi$, then the wealth process $X^\vpi$ and the auxiliary processes $F^{\vpi,(1)}$ and $F^{\vpi,(2)}$, and finally $Z^\vpi$. In each step, we write the perturbed process as a sum of the unperturbed process and an adjustment term that is linear in $\epsilon$ plus higher order terms. The perturbed growth rate is
\begin{align*}
\gamma^{\vpi + \epsilon \vomega}_t &= r_t + (\vpi_t + \epsilon \vomega_t)^\intercal \vtheta_t -  \tfrac{1}{2} (\vpi_t + \epsilon \vomega_t)^\intercal \mSigma_t (\vpi_t + \epsilon \vomega_t)
\\
&= r_t + \vpi_t^\intercal \vtheta_t + \epsilon \, \vomega_t^\intercal \vtheta_t -  \tfrac{1}{2} \Big[ \vpi_t^\intercal \mSigma_t \vpi_t + 2 \epsilon \vomega_t^\intercal \mSigma_t \vpi_t + o(\epsilon) \Big]
\\
&= \gamma^\vpi_t + \epsilon \, \vomega_t^\intercal \left( \vtheta_t - \mSigma_t \vpi_t \right) + o(\epsilon) \, .
\end{align*}
Next, we derive the perturbed wealth process using the perturbed growth rate process obtained above. Substituting, we have
\begin{align*}
d \log X^{\vpi + \epsilon \vomega}_t &= \gamma_t^{\vpi + \epsilon \vomega} \, dt + (\vpi_t + \epsilon \vomega_t)^\intercal \, d\vM_t
\\
&= \Big( \gamma^\vpi_t + \epsilon \, \vomega_t^\intercal \left( \vtheta_t - \mSigma_t \vpi_t \right) + o(\epsilon) \Big) \, dt + \vpi_t^\intercal ~d\vM_t + \epsilon \, \vomega_t^\intercal \, d\vM_t
\\
&= d \log X_t^\vpi + \epsilon ~ \Big\{ \vomega_t^\intercal \left( \vtheta_t - \mSigma_t \vpi_t \right)  dt + \vomega_t^\intercal \, d\vM_t  \Big\} + o(\epsilon) \, .
\end{align*}
Rearranging the equation above and integrating, while noting that $X^{\vpi + \epsilon \vomega}_0 = X_0^\vpi = x$, yields
\[
\log \left(\frac{X^{\vpi + \epsilon \vomega}_t}{X_t^\vpi} \right) - \log \left(\frac{X^{\vpi + \epsilon \vomega}_0}{X_0^\vpi} \right) = \epsilon \left\{ \int_0^t  \vomega_u^\intercal \left( \vtheta_u - \mSigma_u \vpi_u \right)  du + \int_0^t  \vomega_u^\intercal ~d\vM_u \right\} + o(\epsilon) \,.
\]
Now we can write
\begin{align*}
& X^{\vpi + \epsilon \vomega}_t =  X_t^\vpi \exp \big( \epsilon I_t^\vomega + o(\epsilon) \big)
\\
\implies ~~ & X^{\vpi + \epsilon \vomega}_t =  X_t^\vpi \left( 1 + \epsilon I_t^\vomega + o(\epsilon) \right) \, . \numberthis \label{eqn:perturbedWealth}
\end{align*}
Prior to considering the auxiliary processes $F^{\vpi,(1)}$ and $F^{\vpi,(2)}$, we need to write the derivatives of $U\left( X_t^{\vpi + \epsilon \vomega} \right)$ in terms of the unperturbed wealth process. To do so, we use the expression in \eqref{eqn:perturbedWealth} to write
\begin{align*}
U^{(k)} \left( X_t^{\vpi + \epsilon \vomega} \right) &= U^{(k)}\Big(X_t^\vpi + \epsilon X_t^\vpi I_t^\vomega + o (\epsilon) \Big) \, .
\end{align*}
Since $U$ is sufficiently differentiable we can write this expression as a Taylor series around $X_t^\vpi$, namely
\begin{align} \label{eqn:Uderivtaives}
U^{(k)}\Big(X_t^\vpi + \epsilon X_t^\vpi I_t^\vomega + o (\epsilon) \Big) &= U^{(k)}(X_t^\vpi) + \epsilon U^{(k+1)}(X_t^\vpi) X_t^\vpi I_t^\vomega + o (\epsilon) \, .
\end{align}
Using \eqref{eqn:perturbedWealth} and \eqref{eqn:Uderivtaives}, the perturbed auxiliary processes $F^{\vpi + \epsilon \vomega, (k)}$ for $k = 1,2 $ are given by
\begin{align*}
F_t^{\vpi + \epsilon \vomega, (k)} &= U^{(k)}(X_t^{\vpi + \epsilon \vomega}) (X_t^{\vpi + \epsilon \vomega})^k
\\
&= \Big( U^{(k)}(X_t^\vpi) + \epsilon U^{(k+1)}(X_t^\vpi) X_t^\vpi I_t^\vomega + o (\epsilon) \Big) \Big( X_t^\vpi  + \epsilon X_t^\vpi I_t^\vomega + o (\epsilon)  \Big)^k
\\
&= \Big( U^{(k)}(X_t^\vpi) + \epsilon U^{(k+1)}(X_t^\vpi) X_t^\vpi I_t^\vomega + o (\epsilon) \Big) \Big( (X_t^\vpi)^k  + \epsilon k (X_t^\vpi)^k I_t^\vomega + o (\epsilon)  \Big)
\\
&= F_t^{\vpi,(k)} + \epsilon I_t^\vomega \big( k  U^{(k)}(X_t^\vpi)(X_t^\vpi)^k + U^{(k+1)}(X_t^\vpi) (X_t^\vpi)^{k+1} \big) + o(\epsilon)
\\
&= F_t^{\vpi,(k)} + \epsilon I_t^\vomega \left( k  F_t^{\vpi,(k)} + F_t^{\vpi,(k+1)} \right) + o(\epsilon) \, . \numberthis \label{eqn:perturbedF}
\end{align*}
Now we have all the components to find the dynamics of the perturbed utility process $Z^{\vpi + \epsilon \vomega}$. Starting from \eqref{eqn:utilitySDE} and using \eqref{eqn:perturbedWealth} and \eqref{eqn:perturbedF}
{ \small
\begin{align*}
dZ_t^{\vpi + \epsilon \vomega} &= \Big( F_t^{\vpi + \epsilon \vomega, (1)} \big(r_t + ({\vpi_t + \epsilon \vomega_t})^\intercal \vtheta_t \big) + \tfrac{1}{2} F_t^{\vpi + \epsilon \vomega, (2)} ({\vpi_t + \epsilon \vomega_t})^\intercal \mSigma_t ({\vpi_t + \epsilon \vomega_t}) \Big) \, dt
\\
& \qquad \qquad + F_t^{\vpi_t + \epsilon \vomega, (1)} ({\vpi_t + \epsilon \vomega_t})^\intercal \, d\vM_t
\\
&= \bigg( \left[ F_t^{\vpi,(1)} + \epsilon I_t^\vomega \left( F_t^{\vpi,(1)} + F_t^{\vpi,(2)} \right) \right] \big(r_t + \vpi_t^\intercal \vtheta_t + \epsilon \vomega_t^\intercal \vtheta_t \big)
\\
& \qquad \qquad + \tfrac{1}{2} \left[ F_t^{\vpi,(2)} + \epsilon I_t^\vomega \left( 2  F_t^{\vpi,(2)} + F_t^{\vpi,(3)} \right) \right] \Big[ \vpi_t^\intercal \mSigma_t \vpi_t + 2 \epsilon \vomega_t^\intercal \mSigma_t \vpi_t \Big] \bigg) \, dt
\\
& \qquad \qquad + \left[ F_t^{\vpi,(1)} + \epsilon I_t^\vomega \left(  F_t^{\vpi,(1)} + F_t^{\vpi,(2)} \right) \right] ({\vpi_t + \epsilon \vomega_t})^\intercal \, d\vM_t + o(\epsilon)
\\
&= dZ_t^\vpi + \epsilon ~\bigg\{ F_t^{\vpi, (1)} \vomega_t^\intercal \vtheta_t + I_t^\vomega \left( F_t^{\vpi,(1)} + F_t^{\vpi,(2)} \right) (r_t + \vpi_t^\intercal \vtheta_t )
\\
& \qquad \qquad \qquad +  F_t^{\vpi,(2)} \vomega_t^\intercal \mSigma_t \vpi_t  + I_t^\vomega \left(F_t^{\vpi,(2)} + \tfrac{1}{2} F_t^{\vpi,(3)} \right) \vpi_t^\intercal \mSigma_t \vpi_t \Big) \bigg\} \, dt
\\
& \qquad  \qquad  + \epsilon ~\bigg\{ F_t^{\vpi,(1)} \vomega_t^\intercal + I_t^\vomega \left(  F_t^{\vpi,(1)} + F_t^{\vpi,(2)} \right) \vpi_t^\intercal \bigg\} \, d\vM_t + o(\epsilon) \, .
\end{align*}
}%
Since $Z_0^\vpi = Z_0^{\vpi + \epsilon \vomega} = U(x)$, integrating both sides of the equation above from 0 to $T \wedge \tau^{\vpi + \epsilon \vomega}$ and taking expectations yields the desired result provided that the stochastic integral on the RHS of the equation has zero mean, which we prove in the following lemma:
\begin{Lemma} \label{prop:zeroMeanCondition}
For any constrained admissible controls $\vpi, \vomega \in \As_K$ we have $F^{\vpi,(1)}, F^{\vpi,(2)} \in \LL^{\infty,M}_T(\RR)$ and
\[
\EE \left[ \int_0^{T \wedge \tau^{\vpi + \epsilon \vomega}} \bigg\{ F_t^{\vpi,(1)} \vomega_t^\intercal + \left(  F_t^{\vpi,(1)} + F_t^{\vpi,(2)} \right) I_t^\vomega \vpi_t^\intercal \bigg\} \, d\vM_t \right] = 0 \, .
\]
\end{Lemma}
\begin{proof}
The first statement follows from the definition of the constrained admissible set. Namely, the fact that portfolios are stopped at a certain wealth threshold and that the risk-free rate is bounded implies that wealth is bounded on the interval $[0,T]$. Furthermore, since $F^{\vpi, (k)}$ are continuous functions of wealth they must also be bounded on this interval.

Next, we rewrite the stochastic integral as
\[
\int_0^T \ind^{\vpi + \epsilon \vomega}_t \, \bigg\{ F_t^{\vpi,(1)} \vomega_t^\intercal + \left(  F_t^{\vpi,(1)} + F_t^{\vpi,(2)} \right) I_t^\vomega \vpi_t^\intercal \bigg\} \, d\vM_t \, .
\]
The goal is to show that the stochastic integral under the expectation, which we denote $V_t$, is a local martingale. If this is the case then there exists a sequence of stopping times $T_n \uparrow \infty$ a.s. such that $V_{T_n \wedge t}$ is a martingale for each $n$. By choosing $n^* = \inf\left\{ n: T_n > T \right\}$ so that $T_{n^*} \wedge T \ = T$ and $V_{T_{n^*} \wedge t}$ is a martingale which would give
\[
0 = V_0 = V_{T_{n^*} \wedge 0} = \EE\left[ V_{T_{n^*} \wedge T } \right] = \EE[V_T] \,
\]
as required.

To show that $V$ is a local martingale we begin with the following observation: any integral with respect to $\vM$ where the integrand is predictable and in $\LL_T^2(\RR^n)$ is a continuous local martingale by Theorem 30, Ch. IV of \cite{protter2005stochastic}. Applying this to \eqref{eqn:logWealth} it follows that the wealth process $X^\vpi$ has continuous paths. Furthermore, since $U^{(k)}$ is continuous, $F^{\vpi, (k)}$ has continuous paths for all $k \in \NN$ as well. Since $F^{\vpi, (k)}$ is also $\FF$-adapted we can conclude that it is $\FF$-predictable. The indicator is also $\FF$-predictable by the continuity of the wealth paths. Additionally, since $F^{\vpi,(1)} \in \LL^{\infty,M}_T$ and $\vomega \in \LL^2_T(\RR^n)$ is predictable $\int_0^T \ind^{\vpi + \epsilon \vomega}_t \, F_t^{\vpi,(1)} \vomega_t^\intercal \, d\vM_t$ is a (continuous) local martingale.

Next, we show that $W_t \coloneqq \int_0^T \ind^{\vpi + \epsilon \vomega}_t \, \left(  F_t^{\vpi,(1)} + F_t^{\vpi,(2)} \right) I_t^\vomega \vpi_t^\intercal \, d\vM_t$ is also a (continuous) local martingale. By similar reasoning as above $I^\vomega$ is continuous as it is the sum of an ordinary integral and an integral with respect to a continuous martingale, $\vM$, with a predictable integrand, $\vomega$, that is in $\LL^2_T(\RR^n)$. Since $I^\vomega$ is continuous and adapted it is also predictable and hence the integrand in $W$ is predictable. Furthermore, the quadratic variation of $W$ satisfies
\[
\int_0^t \ind^{\vpi + \epsilon \vomega}_t \, (I_t^\vomega)^2 \, \vpi_t^\intercal \mSigma_t \vpi_t \, \, dt ~\leq~ C \bigg( ~\underset{s \in [0,t]}{\sup} (I_s^\vomega)^2 \bigg) \left(\int_0^t \| \vpi_t \|^2 \, dt\right) < \infty ~~ \text{a.s. for all } t \geq 0.
\]
The RHS of the inequality follows because $I^\vomega$ is continuous and hence bounded on compact sets and because $\EE \left[ \int_0^t \| \vpi_t \|^2 \, dt \right] < \infty$ implies that $\int_0^t \| \vpi_t \|^2 \, dt < \infty$ a.s. It follows by Theorem 30, Ch. IV of \cite{protter2005stochastic} that $W_t$ (and hence $V_t$) is a continuous local martingale and the proof is complete.
\end{proof}

This completes the proof of the proposition.
\end{proof} \\

The following result is used to simplify the expression for the G\^{a}teaux derivative, particularly, to handle the term $I_t^\vomega h_t^\vpi$.
\begin{Lemma} \label{lemma:integral1}
Let $a =(a_t)_{t \geq 0}$, $\vb =(\vb_t)_{t \geq 0}$, $\ell =(\ell_t)_{t \geq 0}$ be processes with $a, \ell \in \LL_T^{2}(\RR)$ and $\vb \in \LL_T^{2}(\RR^n)$ and $\FF$-predictable and let $\tau$ be an $\FF$-stopping time with $\tau \leq T$. Then,
{ \small
\begin{align*}
& \EE \left[ \int_0^\tau  \ell_t  \left( \int_0^t  a_u \, du + \int_0^t  \vb_u^\intercal \, d\vM_u\right) dt \right]  = \EE \left[ \int_0^\tau a_t \left( \Ms_t - \int_0^t  \ell_u \, du \right) dt \right] + \EE \left[ \int_0^\tau \vb_t^\intercal \, d\langle \Ms, \vM \rangle_t \right] \, , \numberthis
\end{align*}}%
where $\displaystyle \Ms_t = \EE_t \left[ \int_0^\tau \ell_u \, du \right] \coloneqq \EE \left[ \int_0^\tau \ell_u \, du \, \middle| \, \Ft_t \right]$ and $d\langle \Ms, \vM \rangle$ is a vectorized version of $d\langle \Ms, M^i \rangle$.
\end{Lemma}
\begin{proof}
We treat the two integrals on the LHS of the equation above separately. For the first integral we begin by demonstrating that the integral is finite. Let $\lambda$ denote the Lebesgue measure on $(\RR, \mathcal{B}(\RR))$. Then we have
{\footnotesize
\begin{align*}
\EE\left[ \int_0^\tau \ell_t \left( \int_0^t a_u \, du \right) dt \right] &= \int_{\Omega \times [0,\tau]} \left\{ \ell_t \left( \int_0^t a_u \, du \right) \right\} (\PP \times \lambda)(d\omega, dt)
\\
& \leq \left[\int_{\Omega \times [0,T]} (\ell_t)^2 (\PP \times \lambda)(d\omega, dt)\right] \left[\int_{\Omega \times [0,T]} \left( \int_0^t a_u \, du \right)^2 (\PP \times \lambda)(d\omega, dt)\right]
\\
&\hspace{8.9cm} \text{Cauchy-Schwarz inequality}
\\
&= \EE \left[\int_0^T (\ell_t)^2 \,dt \right] \, \EE \left[\int_0^T \left( \int_0^t a_u \, du \right)^2 dt \right]
\\
&= \EE \left[\int_0^T (\ell_t)^2 \,dt \right] \, \left(\int_0^T \EE \left[ \left( \int_0^t a_u \, du \right)^2 \right] dt\right) \qquad \qquad \text{Tonelli's theorem}
\\
&\leq \EE \left[\int_0^T (\ell_t)^2 \,dt \right] \, \left(\int_0^T \EE \left[ \int_0^t (a_u)^2 \, du \right] dt\right) \hspace{1.75cm} \text{Jensen's inequality}
\\
&< \infty  \hspace{7.9cm} \text{since $a, \ell \in \LL^2_T(\RR)$.}
\end{align*} }
This allows us to change the order of integration by applying Fubini's to write
{\footnotesize \begin{align*}
\EE \Biggl[ \int_0^\tau  \int_0^t \ell_t a_u \,du \,dt \Biggr] &= \EE \Biggl[ \int_0^\tau  \int_u^\tau  \ell_t  a_u \,dt \,du \Biggr] && \text{change order of integration}
\\
&= \EE \Biggl[ \int_0^\tau  \left(\int_u^\tau \ell_t \,dt \right) a_u \,du \Biggr]
\\
&= \EE \Biggl[ \int_0^\tau  \EE_u\left[ \int_u^\tau  \ell_t \,dt \right] a_u \,du \Biggr] && \text{tower property and Fubini's theorem}
\\
&= \EE \Biggl[ \int_0^\tau  \left( \Ms_u - \int_0^u  \ell_t \,dt \right) a_u \,du \Biggr] \, .
\end{align*} }

For the second integral we use a similar argument to show the finiteness of the integral:
{\footnotesize
\begin{align*}
\EE\left[ \int_0^\tau \ell_t \left( \int_0^t \vb_u^\intercal \, d\vM_u \right) dt \right] &= \int_{\Omega \times [0,\tau]} \left\{ \ell_t \left( \int_0^t \vb_u^\intercal \, d\vM_u \right) \right\} (\PP \times \lambda)(d\omega, dt)
\\
& \leq \left[\int_{\Omega \times [0,T]} (\ell_t)^2 (\PP \times \lambda)(d\omega, dt)\right] \left[\int_{\Omega \times [0,T]} \left( \int_0^t \vb_u^\intercal \, d\vM_u \right)^2 (\PP \times \lambda)(d\omega, dt)\right]
\\
&\hspace{8.5cm} \text{Cauchy-Schwarz inequality}
\\
&= \EE \left[\int_0^T (\ell_t)^2 dt \right] \EE \left[\int_0^T \left( \int_0^t \vb_u^\intercal \, d\vM_u \right)^2 dt \right]
\\
&= \EE \left[\int_0^T (\ell_t)^2 dt \right]  \int_0^T \EE \left[ \left( \int_0^t \vb_u^\intercal \, d\vM_u \right)^2 \right] dt \qquad \qquad \text{Tonelli's theorem}
\\
&= \EE \left[\int_0^T (\ell_t)^2 dt \right]  \int_0^T \EE \left[ \int_0^t \vb_u^\intercal \mSigma_u \vb_u \, du \right] \, dt  \qquad \qquad ~~ \text{It\^{o}'s isometry}
\\
&\leq \EE \left[\int_0^T (\ell_t)^2 dt \right] \int_0^T \EE \left[ \int_0^t C \| \vb_u \|^2 \, du \right] \, dt
\\
&< \infty  \hspace{7.75cm} \text{since $\ell \in \LL^2_T(\RR)$, $\vb \in \LL^2_T(\RR^n)$.}
\end{align*} }%
This allows us to once again change the order of integration to write
{ \small
\begin{align*}
\EE \Bigg[ \int_0^\tau \ell_t \left( \int_0^t \vb^\intercal_u \,d\vM_u\right) dt \Bigg] &= \EE \Bigg[ \int_0^\tau  \left(\int_t^\tau \ell_u \,du \right) \vb_t^\intercal \, d\vM_t \Bigg]
\\
&=  \EE \Bigg[ \int_0^\tau   \left( \int_0^\tau \ell_u \,du - \int_0^t \ell_u \,du  \right) \vb_t^\intercal \,d\vM_t \Bigg]
\\
&=  \EE \Bigg[ \int_0^\tau   \left( \Ms_\tau - \int_0^t \ell_u \,du  \right) \vb_t^\intercal \,d\vM_t \Bigg] \, .
\end{align*}
}%
Now we have two terms to consider:
\[
Z_1 = \EE \Bigg[ \int_0^\tau   \Ms_\tau \vb_t^\intercal \,d\vM_t \Bigg] \qquad  \text{and} \qquad Z_2 = \EE \Bigg[ \int_0^\tau  \left( \int_0^t \ell_u \,du  \right) \vb_t^\intercal \,d\vM_t \Bigg] \, .
\]
Denote $L_t \coloneqq \int_0^t \ell_u \,du$ and note that it is continuous in $t$. The integrand  of the stochastic integral appearing in $Z_2$ is predictable since it is the product of a predictable process and a continuous adapted process. The quadratic variation of the stochastic integral is
{\small
\begin{align*}
\int_0^T \ind_{\{\tau \leq T\}}\, L_t^2 \, \vb_t^\intercal \mSigma_t \vb_t \,dt ~&\leq~ C \int_0^T  L_t^2 \, \| \vb_t \|^2 \,dt
\\
&\leq\, C \left( \underset{t \in [0,T]}{\sup} L_t ^2\right) \, \left(\int_0^T \| \vb_t \|^2 \,dt \right)
\\
&< \infty ~~ \text{a.s.} && \text{since $\vb \in \LL^2_T(\RR^n)$ and $L$ is continuous.}
\end{align*}}%
Following the same reasoning used in the proof of Proposition \ref{prop:zeroMeanCondition} we have $Z_2 = 0$. For $Z_1$ we have
{\small \begin{align*}
\EE \Bigg[ \int_0^\tau   \Ms_\tau \vb_t^\intercal \,d\vM_t \Bigg] &= \EE \Bigg[ \Ms_\tau \int_0^\tau    \vb_t^\intercal \,d\vM_t \Bigg]
\\
&= \EE \Bigg[ \left(\int_0^\tau d\Ms_t\right) \left(\int_0^\tau \vb_t^\intercal \,d\vM_t\right) \Bigg]
\\
&= \EE \Bigg[ \int_0^\tau \vb_t^\intercal \,d \langle \Ms, \vM \rangle_t  \Bigg]
\end{align*}}%
The last step follows by It\^{o}'s isometry since both $\vM$ and $\Ms$ are square integrable martingales and $\vb \in \LL_T^2(\RR^n)$.
\end{proof} \\

We are now ready to compute the G\^{a}teaux derivative for our performance criteria.

\begin{Proposition} \label{prop:gateaux}
The functional $H_K: \As_K \rightarrow \RR$ is G\^{a}teaux differentiable for all  $\vpi, \vomega \in \As_K$ with G\^{a}teaux derivative $H_K'(\vpi)$ given by
{\small
\begin{align} \label{eqn:gateaux2}
\left \langle \vomega, H_K'(\vpi) \right\rangle &= \EE \Bigg[ \int_0^{ T \wedge \tau^\vpi } \vomega_t^\intercal \left\{ \left( \vg^\vpi_t + \left( \Ms_t^\vpi - \int_0^t  h^\vpi_u \, du \right) \left( \vtheta_t - \mSigma_t \vpi_t \right) \right) dt + d \left\langle \Ms^\vpi, \vM \right\rangle_t \right\} \Biggr] \, ,
\\
\text{where } \Ms_t^\vpi & \coloneqq \EE_t \left[ \int_0^{T \wedge \tau^\vpi} h_t^\vpi \, dt \right] \text{ is an $\FF$-martingale with } \EE[(\Ms^\vpi_t)^2] < \infty \text{ for all } t. \nonumber
\end{align}}
\end{Proposition}
\begin{proof}
First, notice that
\[
\underset{\epsilon \rightarrow 0}{\lim} ~ \EE \left[ Z^\vpi_{T \wedge \tau^{\vpi + \epsilon \vomega} }\right] = \EE \left[ Z^\vpi_{T \wedge \tau^\vpi} \right] = \EE \left[ Z^\vpi_{T} \right] \, ,
\]
with the last equality following from the fact that portfolios are stopped at $\tau^\vpi$. This allows us to write
\[
\underset{\epsilon \rightarrow 0}{\lim} ~ \frac{H_K(\vpi + \epsilon \vomega) - \EE \left[ Z_{T \wedge \tau^{\vpi + \epsilon \vomega} }\right]}{\epsilon} = \underset{\epsilon \rightarrow 0}{\lim} ~ \frac{H_K(\vpi + \epsilon \vomega) - H_K(\vpi)}{\epsilon} = \left\langle \vomega, H_K'(\vpi) \right\rangle \, .
\]
So by rearranging \eqref{eqn:gateaux1} and taking the limit as $\epsilon$ tends to 0 we find that
\[
\left\langle \vomega, H_K'(\vpi) \right\rangle = \EE \left[ \int_0^{ T \wedge \tau^\vpi } \Big( \vomega_t^\intercal \vg^\vpi_t + I_t^\vomega h^\vpi_t \Big) \, dt \right] \,.
\]
Next, we use the fact that $h$ as well as the integrands in $I_t^\vomega$ of \eqref{eqn:integral} are in $\LL^2_T(\RR)$ to apply Lemma \ref{lemma:integral1} and simplify to obtain the expression in \eqref{eqn:gateaux2}. Finally, notice that since $\Ms^\vpi$ is a Doob martingale and $h \in \LL^2_T(\RR)$ it is in fact a true martingale with finite second moment.
\end{proof} \\

The next lemma gives an explicit representation of $\Ms_t^\vpi - \int_0^t h_u^\vpi \, du$ and $d\langle \Ms^\vpi, \vM \rangle_t$ which will allow us to simplify the G\^{a}teaux derivative and eventually solve the optimal control problem.

\begin{Lemma} \label{lemma:M_process}
Define the processes $q^\vpi = (q^\vpi_t)_{t \geq 0}$ and $Y^\vpi = (Y^\vpi_t)_{t \geq 0}$ as
\begin{align}
q_t^\vpi & \coloneqq F^{\vpi, (1)}_t + F^{\vpi, (2)}_t \, , \text{ and }
\\
Y_t^\vpi & \coloneqq \EE_t \left[ \exp \left( \int_t^{ T \wedge \tau^\vpi } \left[ \tfrac{h_u^\vpi}{F^{\vpi, (1)}_u} - \tfrac{1}{2} \left(\tfrac{q_u^\vpi}{F^{\vpi, (1)}_u}\right)^2 \vpi_u^\intercal \mSigma_u \vpi_u \right] \, du + \int_t^{ T \wedge \tau^\vpi } \tfrac{q_u^\vpi}{F^{\vpi, (1)}_u} \vpi_u^\intercal \, d\vM_u \right) \right] > 0 \, .
\end{align}
Further, write $Y^\vpi$ as the solution to the SDE
\begin{equation}
dY_t^\vpi = Y_t^\vpi \mu_t^\vpi \, dt + Y_t^\vpi (\vsigma_t^\vpi)^\intercal \, d\vM_t \, .
\end{equation}
Then we have the following:
{
\begin{align}
(i) \qquad  &\Ms_t^\vpi - \int_0^t h_u^\vpi \, du = F^{\vpi, (1)}_t \left( Y^\vpi_t - 1 \right)
\\
(ii) \qquad & d \langle \Ms^\vpi, \vM \rangle_t = \mSigma_t \left[ (Y^\vpi_t - 1) q^\vpi_t  \vpi_t + F^{\vpi, (1)}_t Y_t^\vpi \vsigma_t^\vpi \right] dt \hspace{3cm} \label{eqn:quadCov}
\end{align}}
\end{Lemma}
\begin{proof}
To demonstrate the first statement, we apply It\^{o}'s lemma and product rule to obtain
{\small
\begin{align*}
dF^{\vpi, (1)}_t &= \bigg[ \left( F^{\vpi, (1)}_t + F^{\vpi, (2)}_t \right) (r_t + \vpi_t^\intercal \vtheta_t) + \left( F^{\vpi, (2)}_t + \tfrac{1}{2}  F^{\vpi, (3)}_t \right) \vpi_t^\intercal \mSigma_t \vpi_t \bigg] \, dt + \left( F^{\vpi, (1)}_t + F^{\vpi, (2)}_t \right) \vpi_t^\intercal \, d\vM_t
\\
&= h_t^\vpi \, dt + q_t^\vpi \vpi_t^\intercal \, d\vM_t \, .
\end{align*}}%
Next, we write
\[
\frac{dF^{\vpi, (1)}_t}{F^{\vpi, (1)}_t} = \frac{1}{F^{\vpi, (1)}_t} \Big[ h_t^\vpi \, dt + q_t^\vpi \vpi_t^\intercal \, d\vM_t \Big] \, ,
\]
and therefore
\begin{align*}
F^{\vpi, (1)}_{ T \wedge \tau^\vpi } = F^{\vpi, (1)}_t \exp \left( \int_t^{ T \wedge \tau^\vpi } \left[ \tfrac{h_u^\vpi}{F^{\vpi, (1)}_u} - \tfrac{1}{2} \left(\tfrac{q_u^\vpi}{F^{\vpi, (1)}_u}\right)^2 \vpi_u^\intercal \mSigma_u \vpi_u \right] du + \int_t^{ T \wedge \tau^\vpi } \tfrac{q_u^\vpi}{F^{\vpi, (1)}_u} \vpi_u^\intercal ~ d\vM_u \right) \, .
\end{align*}
Then, noting that $\EE_t \left[  \int_t^{ T \wedge \tau^\vpi }  q_u^\vpi \vpi_u^\intercal d\vM_u \right] = 0 $ since $q^\vpi$ is bounded and $\vpi \in \LL^2_T(\RR^n)$ and $\FF$-predictable, we have
{\small
\begin{align*}
\Ms_t^\vpi - \int_0^t h_u^\vpi \, du &= \EE_t \left[  \int_t^{ T \wedge \tau^\vpi }  h_u^\vpi \, du \right]
\\
&= \EE_t \left[  \int_t^{ T \wedge \tau^\vpi }  h_u^\vpi \, du + \int_t^{ T \wedge \tau^\vpi }  q_u^\vpi \vpi_u^\intercal \, d\vM_u \right]
\\
& = \EE_t \left[ \int_t^{ T \wedge \tau^\vpi }  dF^{\vpi, (1)}_u \right]
\\
& = \EE_t \left[ F^{\vpi, (1)}_{ T \wedge \tau^\vpi } - F^{\vpi, (1)}_t  \right]
\\
& = F^{\vpi, (1)}_t \left( Y^\vpi_t - 1 \right) \, ,
\end{align*}}%
which completes the proof of the first statement.

Next,  we are interested in the quadratic covariation process $\langle \Ms^\vpi, \vM \rangle_t$. For this we first write
\begin{align*}
\Ms_t^\vpi
&= \int_0^t h_u^\vpi \, du + F^{\vpi, (1)}_t (Y^\vpi_t  - 1)
\\
\implies \quad d\Ms_t^\vpi &= h_t^\vpi \, dt + \left[d \left( F^{\vpi, (1)}_t (Y^\vpi_t  - 1) \right)\right]
\\
&=  h_t^\vpi \, dt + \left[ (Y^\vpi_t  - 1) dF^{\vpi, (1)}_t + F^{\vpi, (1)}_t dY^\vpi_t + d \langle Y^\vpi, F^{\vpi, (1)} \rangle_t \right]
\\
&= \Big((Y_t^\vpi - 1) q^\vpi_t \vpi_t^\intercal + F^{\vpi, (1)}_t Y^\vpi_t (\vsigma_t^\vpi)^\intercal \Big) \, d\vM_t \numberthis \label{eqn:M_quadVar} \, ,
\end{align*}
where the drift term is zero since $\Ms^\vpi$ is a martingale. This allows us to identify $\langle \Ms^\vpi, \vM \rangle_t$ as the expression given in \eqref{eqn:quadCov}.
\end{proof} \\

Next, we provide a simplified expression for the G\^{a}teaux derivative using the last two results.
\begin{Corollary}
\label{prop:gateaux2}
The G\^{a}teaux derivative given in Proposition \ref{prop:gateaux} can be written as
\begin{align} \label{eqn:gateaux}
\langle \vomega, H_K'(\vpi) \rangle &= \EE \Bigg[ \int_0^{T  \wedge \tau^\vpi} \vomega_t^\intercal Y^\vpi_t \Big( \vg^\vpi_t + F^{\vpi, (1)}_t  \mSigma_t \vsigma_t^\vpi \Big) dt \Biggr] \, .
\end{align}
\end{Corollary}
\begin{proof}
Using Lemma \ref{lemma:M_process} and Proposition \ref{prop:gateaux} we can write the G\^{a}teaux derivative as { \small
\begin{align*}
\langle \vomega, H_K'(\vpi) \rangle &= \EE \Bigg[ \int_0^{T  \wedge \tau^\vpi} \vomega_t^\intercal \left\{  \vg^\vpi_t + F^{\vpi, (1)}_t \left( Y^\vpi_t - 1 \right) \left( \vtheta_t - \mSigma_t \vpi_t \right) + \mSigma_t \left[ (Y^\vpi_t - 1) q^\vpi_t \vpi_t + F^{\vpi, (1)}_t Y^\vpi_t \vsigma_t^\vpi \right] \right\}  dt \Biggr] \, .
\end{align*} }%
Recalling that $\vg^\vpi_t = F_t^{\vpi,(1)} \vtheta_t + F_t^{\vpi,(2)} \mSigma_t \vpi_t$ and $q_t^\vpi = F^{\vpi, (1)}_t + F^{\vpi, (2)}_t$, substituting these terms into the expression above and simplifying yields the result.
\end{proof} \\

The final result we require before deriving the optimal control is a statement concerning the process pair $(Y^\vpi, \vsigma^\vpi)$ that appear in Lemma \ref{lemma:M_process}.

\begin{Lemma} \label{lemma:BSDE}
The pair $(Y^\vpi, \vsigma^\vpi)$ defined in the Lemma \ref{lemma:M_process} satisfy the backward stochastic differential equation (BSDE)
\begin{equation} \label{eqn:BSDE}
\begin{cases}
d \log Y^\vpi_t = -\ind_t^\vpi \Big( \tfrac{h_t^\vpi}{F^{\vpi, (1)}_t} + \tfrac{q_t^\vpi}{F^{\vpi, (1)}_t} (\vsigma_t^\vpi)^\intercal \mSigma_t \vpi_t + \tfrac{1}{2} (\vsigma_t^\vpi)^\intercal \mSigma_t \vsigma_t^\vpi \Big) \, dt + \ind_t^\vpi (\vsigma_t^\vpi)^\intercal \, d\vM_t ~
\\
\log Y^\vpi_T = 0
\end{cases}
\end{equation}
Furthermore, $\vsigma^\vpi \in \LL_T^2(\RR^n)$ and is $\FF$-predictable.
\end{Lemma}
\begin{proof}
We begin by writing
\begin{align*}
Y^\vpi_t &= \EE_t \left[ \exp \left( \int_t^{ T \wedge \tau^\vpi } \left[ \tfrac{h_u^\vpi}{F^{\vpi, (1)}_u} - \tfrac{1}{2} \left(\tfrac{q_u^\vpi}{F^{\vpi, (1)}_u}\right)^2 \vpi_u^\intercal \mSigma_u \vpi_u \right] \, du + \int_t^{ T \wedge \tau^\vpi } \tfrac{q_u^\vpi}{F^{\vpi, (1)}_u} \vpi_u^\intercal \, d\vM_u \right) \right]
\\
&= \EE_t \left[ \exp \left( \int_t^T \ind_u^\vpi \left[ \tfrac{h_u^\vpi}{F^{\vpi, (1)}_u} - \tfrac{1}{2} \left(\tfrac{q_u^\vpi}{F^{\vpi, (1)}_u}\right)^2 \vpi_u^\intercal \mSigma_u \vpi_u \right] \, du + \int_t^T \ind_u^\vpi \tfrac{q_u^\vpi}{F^{\vpi, (1)}_u} \vpi_u^\intercal \, d\vM_u \right) \right]
\\
& = \EE_t \left[ \frac{\Gamma_T}{\Gamma_t} \right] \qquad
\\
\implies \qquad \Gamma_t Y^\vpi_t &= \EE_t \left[ \Gamma_T \right] \, ,
\end{align*}
where $\Gamma$ satisfies the SDE
\[
d \Gamma_t = \Gamma_t \ind_t^\vpi \Big( \tfrac{h_u^\vpi}{F^{\vpi, (1)}_u} \, dt + \tfrac{q_t^\vpi}{F^{\vpi, (1)}_t} \vpi_t^\intercal \, d\vM_t \Big) ~, \qquad \Gamma_0 = 1 \, .
\]
Also, since $Y^\vpi$ stops once $\tau^\vpi$ is reached we may write
\[
dY^\vpi_t = \ind_t^\vpi Y^\vpi_t \mu_t^\vpi \, dt + \ind_t^\vpi Y^\vpi_t (\vsigma_t^\vpi)^\intercal \, d\vM_t ~, \qquad Y^\vpi_T = 1 \, .
\]
To find $\mu^\vpi$ we apply It\^{o}'s product rule to obtain
{ \footnotesize
\begin{align*}
d(Y^\vpi_t \Gamma_t) &= Y^\vpi_t d\Gamma_t  + \Gamma_t dY^\vpi_t  + d[\Gamma, Y^\vpi]_t
\\
&= Y^\vpi_t \Gamma_t \ind_t^\vpi \Big( \tfrac{h_t^\vpi}{F^{\vpi, (1)}_t} \, dt + \tfrac{q_t^\vpi}{F^{\vpi, (1)}_t} \vpi_t^\intercal d\vM_t \Big) + \Gamma_t \ind_t^\vpi \Big( Y^\vpi_t \mu_t^\vpi \, dt + Y^\vpi_t (\vsigma_t^\vpi)^\intercal \, d\vM_t \Big)+  Y^\vpi_t \Gamma_t \ind_t^\vpi \tfrac{q_t^\vpi}{F^{\vpi, (1)}_t} (\vsigma_t^\vpi)^\intercal \mSigma_t \vpi_t \, dt
\\
&= Y_t^\vpi \Gamma_t \Big( \ind_t^\vpi \mu_t^\vpi + \ind_t^\vpi \tfrac{h_t^\vpi}{F^{\vpi, (1)}_t} + \ind_t^\vpi \tfrac{q_t^\vpi}{F^{\vpi, (1)}_t} (\vsigma_t^\vpi)^\intercal \mSigma_t \vpi_t \Big) \, dt + \ind_t^\vpi  Y_t^\vpi \Gamma_t \Big(  \vsigma_t^\vpi + \tfrac{q_t^\vpi}{F^{\vpi, (1)}_t} \vpi_t \Big)^\intercal \, d\vM_t \, .
\end{align*} }%
Since $Y^\vpi \Gamma$ is a martingale, the drift term in the SDE above must be equal to zero. Therefore, $\mu^\vpi$ is given by
\[
\mu_t^\vpi = - \Big( \tfrac{h_t^\vpi}{F^{\vpi, (1)}_t} + \tfrac{q_t^\vpi}{F^{\vpi, (1)}_t} (\vsigma_t^\vpi)^\intercal \mSigma_t \vpi_t \Big) \, .
\]
Substituting back into the SDE satisfied by $Y^\vpi$ and applying It\^{o}'s lemma yields the result.

Next, we have that
{\footnotesize
\begin{align*}
\EE \left[ \int_0^T (Y_t^\vpi)^2 \, dt \right] &= \EE \left[ \int_0^T \left\{ \EE_t \left[ \exp \left( \int_t^{ T \wedge \tau^\vpi } \left[ \tfrac{h_u^\vpi}{F^{\vpi, (1)}_u} - \tfrac{1}{2} \left(\tfrac{q_u^\vpi}{F^{\vpi, (1)}_u}\right)^2 \vpi_u^\intercal \mSigma_u \vpi_u \right] \, du + \int_t^{ T \wedge \tau^\vpi } \tfrac{q_u^\vpi}{F^{\vpi, (1)}_u} \vpi_u^\intercal \, d\vM_u \right) \right] \right\}^2 \, dt \right]
\\
& \leq \EE \left[ \int_0^T \EE_t \left[ \exp \left( 2 \int_t^{ T \wedge \tau^\vpi }  \left[ \tfrac{h_u^\vpi}{F^{\vpi, (1)}_u} - \tfrac{1}{2} \left(\tfrac{q_u^\vpi}{F^{\vpi, (1)}_u}\right)^2 \vpi_u^\intercal \mSigma_u \vpi_u \right] \, du + 2 \int_t^{ T \wedge \tau^\vpi } \tfrac{q_u^\vpi}{F^{\vpi, (1)}_u} \vpi_u^\intercal \, d\vM_u \right) \right] \, dt \right]
\\
& = \int_0^T \EE \left[  \EE_t \left[ \exp \left( 2 \int_t^{ T \wedge \tau^\vpi }  \left[ \tfrac{h_u^\vpi}{F^{\vpi, (1)}_u} - \tfrac{1}{2} \left(\tfrac{q_u^\vpi}{F^{\vpi, (1)}_u}\right)^2 \vpi_u^\intercal \mSigma_u \vpi_u \right] \, du + 2 \int_t^{ T \wedge \tau^\vpi } \tfrac{q_u^\vpi}{F^{\vpi, (1)}_u} \vpi_u^\intercal \, d\vM_u \right) \right] \right] \, dt
\\
& = \int_0^T \EE \left[ \exp \left( 2 \int_t^{ T \wedge \tau^\vpi }  \left[ \tfrac{h_u^\vpi}{F^{\vpi, (1)}_u} - \tfrac{1}{2} \left(\tfrac{q_u^\vpi}{F^{\vpi, (1)}_u}\right)^2 \vpi_u^\intercal \mSigma_u \vpi_u \right] \, du + 2 \int_t^{ T \wedge \tau^\vpi } \tfrac{q_u^\vpi}{F^{\vpi, (1)}_u} \vpi_u^\intercal \, d\vM_u \right) \right] \, dt
\\
& = \int_0^T \EE \left[ \exp \left( 2 \int_t^{ T \wedge \tau^\vpi }  \left[ \tfrac{h_u^\vpi}{F^{\vpi, (1)}_u} - \tfrac{1}{2} \left(\tfrac{q_u^\vpi}{F^{\vpi, (1)}_u}\right)^2 \vpi_u^\intercal \mSigma_u \vpi_u \right] \, du + 2 \int_t^{ T \wedge \tau^\vpi } \tfrac{q_u^\vpi}{F^{\vpi, (1)}_u} \vpi_u^\intercal \, d\vM_u \right) \right] \, dt
\\
& = \int_0^T \EE \left[ \exp \left( \int_t^{ T \wedge \tau^\vpi }  \tfrac{2 h_u^\vpi}{F^{\vpi, (1)}_u} \right) \, \mathcal{E} \left(  \int_t^{ T \wedge \tau^\vpi }  \tfrac{2 q_u^\vpi}{F^{\vpi, (1)}_u} \vpi_u^\intercal \, d\vM_u \right) \right] \, dt
\\
& < \infty
\end{align*}
}%
since the first and second integrands in the last line above are bounded and in $\LL^2_T(\RR^n)$, respectively. This implies that $\log Y_T^\vpi \in \LL^2_T(\RR)$ by Jensen's inequality. However, we also have that
\begin{align*}
\log Y_t^\vpi = \int_0^t \left(\mu_u^\vpi - \tfrac{1}{2} \langle Y^\vpi , Y^\vpi \rangle_u \right) du + \int_0^t (\vsigma_u^\vpi)^\intercal d\vM_u \, ,
\end{align*}
which is in $\LL^2_T(\RR)$ and $\FF$-adapted only if $\vsigma^\vpi \in \LL^2_T(\RR^n)$ and $\FF$-predictable.
\end{proof} \\

\subsection{Optimality}

We proceed to finding the optimal control for the stochastic control problem \eqref{eqn:control}. To this end, we use the results of the previous section to find the \textit{unique} control that causes the G\^{a}teaux derivative to vanish and relate it to the solution of a FBSDE. We begin by providing a necessary and sufficient condition for the G\^{a}teaux derivative to vanish in the constrained problem.
\begin{Proposition} \label{prop:gat_vanish}
The G\^{a}teaux derivative \eqref{eqn:gateaux} vanishes in all directions, i.e. $\langle \vomega , H_K'(\vpi) \rangle = 0$ for all $\vomega \in \As_K$, if and only if $\vg^\vpi_t + F^{\vpi, (1)}_t  \mSigma_t \vsigma_t^\vpi = 0$ $t$-a.e. in the interval $[0,{T  \wedge \tau^\vpi}]$, $\PP$-a.s..
\end{Proposition}
\begin{proof}
The G\^{a}teaux derivative is:
\begin{align*}
\langle \vomega, H_K'(\vpi) \rangle &= \EE \Bigg[ \int_0^{T  \wedge \tau^\vpi} \vomega_t^\intercal Y^\vpi_t \Big( \vg^\vpi_t + F^{\vpi, (1)}_t  \mSigma_t \vsigma_t^\vpi \Big) \, dt \Biggr] \, .
\end{align*}
Clearly, if $\vg^\vpi_t + F^{\vpi, (1)}_t  \mSigma_t \vsigma_t^\vpi = 0$ then $\langle \vomega , H_K'(\vpi) \rangle = 0$ for all $\vomega \in \As_K$.

We prove necessity by contradiction. Assume that $\langle \vomega , H_K'(\vpi) \rangle = 0$ for all $\vomega \in \As_K$. Assume further that
\[
B = \left\{ (\omega, t) \in \Omega \times [0,T \wedge \tau^\vpi] : \left(\vg^\vpi_t + F^{\vpi, (1)}_t  \mSigma_t \vsigma_t^\vpi\right)(\omega) \neq 0 \right\}
\]
has positive measure. Now define the process
\[
\vomega_t = \left[\left(\vg^\vpi_t + F^{\vpi, (1)}_t  \mSigma_t \vsigma_t^\vpi\right) \ind_B\right] \ind^\vomega_t
\]
This is an admissible portfolio since it is $\FF$-predictable, $F^{\vpi, (1)}$ is bounded, $\vg^\vpi, \vsigma^\vpi \in \LL^2_T(\RR^n)$ and the process is stopped once the wealth threshold is reached. It follows that
\begin{align*}
\langle \vomega, H_K'(\vpi) \rangle &= \EE \Bigg[ \int_0^{T  \wedge \tau^\vpi \wedge \tau^\vomega } Y^\vpi_t \, \vomega_t^\intercal \vomega_t \ind_B \, dt \Biggr] > 0 \,
\end{align*}
since $Y^\vpi_t > 0$ and $\vomega_t^\intercal \vomega_t > 0$ on the set $B$. This gives our contradiction and hence $B$ must have zero measure, which completes the proof.
\end{proof} \\

We now present our main theorem which characterizes the optimal portfolio for our stochastic control problem.
\begin{Theorem} \label{thm:optCont}
Define the processes $\zeta = \left( \zeta_t \right)_{t \geq 0}$ and $\phi = \left( \phi_t \right)_{t \geq 0}$ by
\begin{align}
\zeta_t &= -\frac{F_t^{\vpi,(1)} }{F_t^{\vpi,(2)} } \, ,
\qquad \qquad
\phi_t = \frac{F_t^{\vpi,(3)} }{F_t^{\vpi,(2)} } \, ,
\end{align}
and the portfolio process $\vpi^* = (\vpi_t^*)_{t\geq 0}$ by
\begin{equation} \label{eqn:optCont}
\vpi_t^* = \zeta_t \left( \mSigma_t^{-1} \vtheta_t + \vsigma_t^{\vpi^*} \right) \;,
\end{equation}
along with the FBSDE
{\small
\begin{equation} \label{eqn:FBSDE}
\begin{cases}
d X^{\vpi^*}_t  = X^{\vpi^*}_t \big(r_t + \zeta_t \vtheta_t^\intercal \mSigma^{-1}_t \vtheta_t +  \zeta_t \vtheta_t^\intercal \vsigma^{\vpi^*}_t \big) \, dt + X^{\vpi^*}_t \zeta_t  \left( \mSigma_t^{-1} \vtheta_t + \vsigma_t^{\vpi^*} \right)^\intercal d\vM_t\\
X_0^{\vpi^*} = x
\\
d \log Y^{\vpi^*}_t = \Big( A_t r_t + B_t \vtheta^\intercal \mSigma^{-1}_t \vtheta_t + B_t \vtheta_t^\intercal \vsigma^{\vpi^*}_t + \left( B_t +\tfrac{1}{2} \zeta_t \phi_t \right) (\vsigma_t^{\vpi^*})^\intercal \mSigma_t \vsigma_t^{\vpi^*} \Big) \, dt + (\vsigma_t^{\vpi^*})^\intercal \, d\vM_t ~
\\
\log Y^{\vpi^*}_T = 0
\end{cases}
\end{equation}
}%
where $A_t = \tfrac{1}{\zeta_t} - 1$ and $B_t = 1 + \tfrac{1}{2} \zeta_t \phi_t$.

If one of the two following conditions holds:
\begin{enumerate}[label=(\alph*),labelsep=0.5cm]
	\item $\zeta \in \LL^{\infty,M}_T(\RR)$, or
	\item $\zeta_t X_t^{\vpi^*}$ is $\PP$-a.s. continuous with $\EE\left[ \underset{t \in [0,T]}{\sup} ~(\zeta_t X_t^{\vpi^*})^2 \right] < \infty$ and the wealth equation \eqref{eqn:dollarSDE} corresponding to $\widetilde{\vpi}_t^* = \vpi_t^* X_t^{\vpi^*}$ has a unique square-integrable solution with $\EE [(X_T^{\vpi^*})^2] < \infty$,
\end{enumerate}
then $\vpi^*$ is an admissible portfolio and is the unique solution to the stochastic control problem \eqref{eqn:control} and, furthermore, the FBSDE \eqref{eqn:FBSDE} has a unique solution.
\end{Theorem}
\begin{Remark}
The process $\frac{1}{\zeta_t} = -\frac{F_t^{\vpi,(2)} }{F_t^{\vpi,(1)}} = -\frac{U''(X_t^\vpi) (X_t^\vpi)}{U'(X_t^\vpi)} $ is the Arrow–Pratt measure of relative  risk aversion (or coefficient of relative risk aversion).
\end{Remark}
\begin{Remark}
The BSDE component of \eqref{eqn:FBSDE} is a quadratic BSDE. The existence and uniqueness of solutions to BSDEs of this type are discussed in \cite{kobylanski2000backward} when the noise process is a Brownian motion and in \cite{morlais2009quadratic} for the case of more general martingale noise processes. When the equations are fully coupled the reader is referred to \cite{luo2015solvability} and \cite{luo2017solvability} for existence and uniqueness results.
\end{Remark}
\begin{proof}
First, we show that $\vpi^*$ is an admissible portfolio for the unconstrained problem \eqref{eqn:control} and hence its stopped counterpart $\vpi_K^*$ is an admissible portfolio for the constrained problem \eqref{eqn:constrained_control}. Since all the processes appearing in \eqref{eqn:optCont} are either continuous and adapted or predictable, we have that $\vpi^*$ is predictable under either condition (a) or (b). Next, recall that $\vtheta$ is bounded and, from Appendix A of \cite{alaradi2018}, the elements of $\mSigma^{-1}$ are also bounded and we have that $\vsigma^\vpi \in \LL^2_T(\RR^n)$. Thus, when $\zeta$ is bounded $\vpi^* \in \LL^2_T(\RR^n)$ and is clearly an admissible portfolio. To show admissibility under the alternative condition on $\zeta$ we work with the auxiliary control problem where the control process is given in terms of dollar amounts as in the proof of Proposition \ref{prop:concave}. The dollar amount control process corresponding to the weight process \eqref{eqn:optCont} is
\[
\widetilde{\vpi}^*_t = \zeta_t X_t^{\vpi^*} \left( \mSigma_t^{-1} \vtheta_t + \vsigma_t^{\vpi^*} \right)
\]
This process is locally square-integrable, i.e. $\int_0^T \| \widetilde{\vpi}^*_t \|^2 < \infty$ $\PP$-a.s., due to the assumptions made in condition (b). Therefore, by Lemma 3.1 of \cite{lim2004quadratic}, $\widetilde{\vpi}^*$ is an admissible \textit{dollar amount} process and hence $\vpi^*$ must be an admissible weight process.

Next, it is easy to verify that ${\vpi_K^*}$ is the unique portfolio that satisfies
\[
\vg^{\vpi_K^*}_t + F^{{\vpi_K^*}, (1)}_t  \mSigma_t \vsigma_t^{\vpi_K^*} = 0 \, ~~ \text{ \textit{on} } ~~ [0, T \wedge \tau^\vpi].
\]
This in turn implies by Proposition \ref{prop:gat_vanish} that ${\vpi_K^*}$ is the unique portfolio at which the G\^{a}teaux derivative \eqref{eqn:gateaux} vanishes in all directions, i.e. $\langle \vomega_K , H_K'(\vpi^*) \rangle = 0$ for all $\vomega_K \in \As_K$. Now, we take the limit as $K$ tends to infinity to conclude that ${\vpi^*}$ is the unique portfolio such that $\langle \vomega , H'(\vpi^*) \rangle = 0$ for all $\vomega \in \As$. To do this we consider the limit
\[
\underset{K \rightarrow \infty}{\lim} ~ \underset{n \rightarrow \infty}{\lim} \, \frac{H_K(\vpi^*_K + \epsilon_n \vomega_K) - H_K(\vpi^*_K)}{\epsilon_n}
\]
where $\left\{ \epsilon_n \right\}_{n\in\NN}$ is a sequence of real numbers tending to zero. Denoting the ratio in the expression above by $a_{nk}$, we note that we can interchange the order of the limits if $\underset{K \rightarrow \infty}{\lim} \, a_{nk} $ exists for all fixed $n$ and $\underset{n \rightarrow \infty}{\lim} \, a_{nk} $ exists for all fixed $K$. The latter is true since $\underset{n \rightarrow \infty}{\lim} \, a_{nk} = \langle \vomega_K , H_K'(\vpi^*) \rangle = 0 $. For the other limit notice that $\vpi_K^* \rightarrow \vpi^*$ in the $\LL^2$ norm as $K \rightarrow \infty$ and so $H_K(\vpi_K^*) \rightarrow H(\vpi^*)$ as $K \rightarrow \infty$ by the fact that $H_K(\vpi^*_K) = H(\vpi_K^*)$ and the continuity of $H$. Therefore,
\begin{align*}
\underset{K \rightarrow \infty}{\lim} \, a_{nk} &= \underset{K \rightarrow \infty}{\lim} \, \frac{H_K(\vpi^*_K + \epsilon_n \vomega_K) - H_K(\vpi^*_K)}{\epsilon_n}
\\
&= \underset{K \rightarrow \infty}{\lim} \, \frac{H(\vpi^*_K + \epsilon_n \vomega_K) - H(\vpi^*_K)}{\epsilon_n}
\\
&= \frac{H(\vpi^* + \epsilon_n \vomega) - H(\vpi^*)}{\epsilon_n}
\end{align*}
is a well-defined limit and thus we can interchange the order of the limits to obtain
\[
\underset{K \rightarrow \infty}{\lim} ~ \underset{n \rightarrow \infty}{\lim} \, \frac{H_K(\vpi^*_K + \epsilon_n \vomega_K) - H_K(\vpi^*_K)}{\epsilon_n} = \underset{n \rightarrow \infty}{\lim} \frac{H(\vpi^* + \epsilon_n \vomega) - H(\vpi^*)}{\epsilon_n} = \langle \vomega , H'(\vpi^*) \rangle \, .
\]
This quantity must be equal to zero since the inner limit in the LHS of the equation above is equal to 0 for each $K \in \NN$. Since $\vpi^*$ is the only stationary point for the functional $H$ and since we have shown that there exists a unique global maximizer, which must also be a stationary point for $H$, ${\vpi^*}$ must be the unique global maximizer.

Finally, the FBSDE \eqref{eqn:FBSDE} is obtained by substituting the optimal control process \eqref{eqn:optCont} in the forward SDE \eqref{eqn:portSDE} and the limiting version of the BSDE \eqref{eqn:BSDE} when $K \rightarrow \infty$ and simplifying. The existence and uniqueness of the optimal control and its dependence on $\vsigma^{\vpi^*}$ in turn implies that the FBSDE \eqref{eqn:FBSDE} must have a unique solution, which completes the proof.
\end{proof}

\section{Specific Utility Functions}

\subsection{Logarithmic Utility - $U(x) = \log x$}

It is easy to check that $F^{\vpi,(k)}_t = (-1)^{k+1} (k-1)!$ from which we can conclude that $\zeta_t = 1$. Since this process is bounded, we may apply Theorem \ref{thm:optCont} to derive the optimal control. Furthermore, the processes $A_t$ and $B_t$ given in Theorem \ref{thm:optCont} vanish and hence the pair $Y^{\vpi^*}_t = 1$ and $\vsigma_t^\vpi = 0$ solve the BSDE in \eqref{eqn:FBSDE}, which is now decoupled from the forward SDE. The optimal portfolio is therefore given by
\[
\vpi_t^* = \mSigma_t^{-1} \vtheta_t \,.
\]
Note that this portfolio is the well-known \textit{growth optimal portfolio}.


\subsection{Power Utility - $U(x) = \frac{x^\eta}{\eta}, ~ \eta < 1$}
In this case, we can check that $F^{\vpi,(k)}_t = (X_t^\vpi)^\eta \prod_{i=1}^{k-1} (\eta - i)  $ and therefore $\zeta_t = \frac{1}{1-\eta}$ is once again bounded. Substituting the relevant expressions in the FBSDE \eqref{eqn:FBSDE} gives the decoupled system
{\small \begin{equation}
\begin{cases}
d X^{\vpi^*}_t  = X^{\vpi^*}_t \big(r_t + \tfrac{1}{1-\eta} \vtheta_t^\intercal \mSigma^{-1}_t \vtheta_t +  \tfrac{1}{1-\eta} \vtheta_t^\intercal \vsigma^{\vpi^*}_t \big) dt + \tfrac{1}{1-\eta} X^{\vpi^*}_t \left( \mSigma_t^{-1} \vtheta_t + \vsigma_t^{\vpi^*} \right)^\intercal d\vM_t \\
X_0^{\vpi^*} = x
\\
d \log Y^{\vpi^*}_t = -\Big( \eta r_t + \tfrac{1}{2} \tfrac{\eta}{1-\eta} \vtheta_t^\intercal \mSigma_t^{-1} \vtheta_t + \tfrac{\eta}{1-\eta} \vtheta_t^\intercal \vsigma_t^{\vpi^*}  + \tfrac{1}{2} \tfrac{1}{1-\eta} (\vsigma_t^{\vpi^*})^\intercal \mSigma_t \vsigma_t^{\vpi^*} \Big)  ~dt + (\vsigma_t^{\vpi^*})^\intercal ~d\vM_t ~
\\
\log Y^{\vpi^*}_T = 0
\end{cases}
\end{equation}}%
Next, notice that if we apply It\^{o}'s lemma to rewrite the BSDE in geometric form this gives
\begin{equation*}
\begin{cases}
\frac{d Y^{\vpi^*}_t}{Y^{\vpi^*}_t} = -\Big( \eta r_t + \tfrac{1}{2} \tfrac{\eta}{1-\eta} \vtheta_t^\intercal \mSigma_t^{-1} \vtheta_t + \tfrac{\eta}{1-\eta} \vtheta_t^\intercal \vsigma_t^{\vpi^*}  + \tfrac{1}{2} \tfrac{\eta}{1-\eta} (\vsigma_t^{\vpi^*})^\intercal \mSigma_t \vsigma_t^{\vpi^*} \Big)  ~dt + (\vsigma_t^{\vpi^*})^\intercal ~d\vM_t ~
\\
Y_T^{\vpi^*}= 1
\end{cases}
\end{equation*}
which is identical to the BSDE derived for the power utility case in \cite{ferland2008fbsde} when the noise process is a Brownian motion; with $Y^{\vpi^*}_t$ playing the role of $p_t$ and $\vsigma_t^{\vpi^*}$ replacing $\mSigma^{-1/2}_t \frac{\Lambda_t}{p_t}$. Therefore, we can lean on the results obtained in that paper, in particular, the form of the optimal solution when the model parameters are deterministic.

\subsection{Exponential Utility - $U(x) = -e^{-\gamma x}$}
In this case, we have $F^{\vpi,(k)}_t = (-1)^{k+1} \gamma^k (X_t^\vpi)^k e^{-\gamma X_t^\vpi}$ and therefore $\zeta_t = \frac{ 1 }{ \gamma X_t^\vpi }$ and $\zeta_t \phi_t = 1$. In this case, we find that $\zeta_t X_t^\vpi = \frac{ 1 }{ \gamma }$ is constant. Moreover, the FBSDE in \eqref{eqn:FBSDE} reduces to
{\small
\begin{equation}
\begin{cases}
d X^{\vpi^*}_t  =  \big(r_t X^{\vpi^*}_t + \tfrac{1}{\gamma} \vtheta_t^\intercal \mSigma^{-1}_t \vtheta_t +  \tfrac{1}{\gamma} \vtheta_t^\intercal \vsigma^{\vpi^*}_t \big) dt + \tfrac{1}{\gamma}  \left( \mSigma_t^{-1} \vtheta_t + \vsigma_t^{\vpi^*} \right)^\intercal d\vM_t
\\
X_0^{\vpi^*} = x
\\
d \log Y^{\vpi^*}_t = -\Big( (1 - \gamma X_t^{\vpi^*}) r_t - \tfrac{1}{2} \vtheta_t^\intercal \mSigma_t^{-1} \vtheta_t - \vtheta_t^\intercal \vsigma_t^{\vpi^*}  \Big)  ~dt + (\vsigma_t^{\vpi^*})^\intercal ~d\vM_t ~
\\
\log Y^{\vpi^*}_T = 0
\end{cases}
\end{equation}
}%
and it is not difficult to verify that the forward SDE has a unique square-integrable solution so that condition (b) of Theorem \ref{thm:optCont} is satisfied.
Now, if we define a new process by $\log \widetilde{Y}_t^{\vpi^*} = \log Y_t^{\vpi^*} + \int_t^T r_s ~ds$, we find that FBSDE system for $X^{\vpi^*}$ and $\widetilde{Y}^{\vpi^*}$ is
\begin{equation}
\begin{cases}
d X^{\vpi^*}_t  =  \big(r_t X^{\vpi^*}_t + \tfrac{1}{\gamma} \vtheta_t^\intercal \mSigma^{-1}_t \vtheta_t +  \tfrac{1}{\gamma} \vtheta_t^\intercal \vsigma^{\vpi^*}_t \big) dt + \tfrac{1}{\gamma}  \left( \mSigma_t^{-1} \vtheta_t + \vsigma_t^{\vpi^*} \right)^\intercal d\vM_t
\\
X_0^{\vpi^*} = x
\\
d \log \widetilde{Y}^{\vpi^*}_t = \Big( \gamma X_t^{\vpi^*} r_t + \tfrac{1}{2} \vtheta_t^\intercal \mSigma_t^{-1} \vtheta_t + \vtheta_t^\intercal \vsigma_t^{\vpi^*}  \Big)  ~dt + (\vsigma_t^{\vpi^*})^\intercal ~d\vM_t ~
\\
\log \widetilde{Y}^{\vpi^*}_T = 0
\end{cases}
\end{equation}
which can be mapped to Equation (16) of \cite{ferland2008fbsde} when the noise process is a Brownian motion. Once again we refer the reader to that paper for detailed results.

\section{Conclusions}

In this paper we have solved the Merton problem of maximizing the expected utility of terminal wealth using variational analysis techniques, providing sufficient \textit{and necessary} conditions for the optimal solution in terms of the solution to a quadratic FBSDE. One extension in which this approach would be useful is in the setting of partial information where the noise drivers of asset prices are general martingales. In this case, the usual approach using the stochastic maximum principle would not lead to a solution. Some of the present results would also be useful in extending the results of the outperformance and tracking problem discussed in \cite{alaradi2018} and \cite{alaradi2019active} to the partial information setting where more general utility functions are considered. It should also be possible to extend the results of the current paper to the case where the noise processes may jump.

\bibliographystyle{chicago}
\bibliography{merton_problem_convex_analysis}

\end{document}